\documentclass[conference,letterpaper]{IEEEtran}

\usepackage[utf8]{inputenc}
\usepackage[T1]{fontenc}
\usepackage{url}
\usepackage{ifthen}
\usepackage{cite}
\usepackage[cmex10]{amsmath}

\usepackage{graphicx}
\usepackage{amsfonts}
\usepackage{amsmath}
\usepackage{amssymb}
\usepackage{mathrsfs} 
\usepackage{color}
\usepackage{bm}
\usepackage{latexsym} 
\newtheorem{thm}{Theorem}

\newtheorem{lem}{Lemma}
\newtheorem{proof}{proof}

\newtheorem{defn}{Definition}
\newtheorem{rem}{Remark}
\newtheorem{exam}{Example}

\begin{document}

\title{Coded Caching with Polynomial Subpacketization}

\author{\IEEEauthorblockN{Wentu Song,
Kui Cai, and Long Shi} \IEEEauthorblockA{Science and Math Cluster,
Singapore University of Technology and Design, Singapore\\ Email:
\{wentu$\_$song, cai$\_$kui, shi$\_$long\}@sutd.edu.sg}}
\maketitle


\begin{abstract}
Consider a centralized caching network with a single server and
$K$ users. The server has a database of $N$ files with each file
being divided into $F$ packets ($F$ is known as subpacketization),
and each user owns a local cache that can store $\frac{M}{N}$
fraction of the $N$ files. We construct a family of centralized
coded caching schemes with polynomial subpacketization.
Specifically, given $M$, $N$ and an integer $n\geq 0$, we
construct a family of coded caching schemes for any $(K,M,N)$
caching system with $F=O(K^{n+1})$. More generally, for any
$t\in\{1,2,\cdots,K-2\}$ and any integer $n$ such that $0\leq
n\leq t$, we construct a coded caching scheme with
$\frac{M}{N}=\frac{t}{K}$ and $F\leq
K\binom{\left(1-\frac{M}{N}\right)K+n}{n}$.
\end{abstract}


\section{Introduction}

A $(K,M,N)$ caching system consists of one server and $K$ users,
where all users connect to the server through a shared, error-free
link. The server has a database of $N$ files and each user may
request a specific file from the server at certain time in the
future. The user requests are random and not known by the server
in advance. Each user has a cache that can store $M/N$ fraction of
the $N$ files of the server. A \emph{centralized coded caching
scheme} operates in two separated phases: the \emph{placement
phase} and the \emph{delivery phase}. In the placement phase, the
server allocates certain packets of the data files into the cache
of the users, while in the delivery phase, the server, upon
receiving the specific demands of all users, broadcasts coded
packets through the shared link to all users so that each user can
extract its requested file from the received packets and its cache
content. The rate $R$ of the scheme is defined as the maximal
transmission amount in the delivery phase among all possible
combinations of the user demands, and the primary goal is to
design coded caching scheme with as small rate as possible.

Coded caching problem was first investigated by Maddah-Ali and
Niesen in their award-winning paper \cite{Maddah-Ali14}. The coded
caching scheme proposed in \cite{Maddah-Ali14} attains the rate
\begin{align}\label{MN-rate}R^*=\frac{K\left(1-\frac{M}{N}\right)}
{1+K\frac{M}{N}},\end{align} where $1-\frac{M}{N}$ is called the
local caching gain and $1+K\frac{M}{N}$ is called the global
caching gain, and $R^*$ was proved to be optimal among schemes
with uncoded placement \cite{Wan16,Yu17}.

A major limitation of the Maddah-Ali-Niesen scheme is the
exponential subpacketization problem: by this caching scheme, each
file is divided into $F=\binom{K}{KM/N}$ packets $(F$ is referred
to as the file size or subpacketization.$)$, which grows
exponentially with $K$ \cite{Shanmugam16}. Since high
subpacketization may result in transmission delay in practical
implementations, coded caching with low subpacketization,
especially polynomial subpacketization, is of great interest.

Many works have been engaged to reduce the subpacketization, with
the sacrifice of increasing the rate. A user-grouping method was
adopted in \cite{Shanmugam16} to reduce the subpacketization
level, and a more general concatenating construction method was
used in \cite{Cheng19-1}. A framework of constructing centralized
coded caching scheme, named placement delivery array design (or
PDA design for simplicity), was introduced in \cite{Yan17}, based
on which some new classes of coded caching schemes were obtained
in \cite{Yan17} and \cite{Cheng19}. Caching schemes constructed
using other techniques, such as hypergraphs, bipartite graphs
combinatorial designs, and projective geometries over finite
fields, are reported in \cite{Chong18}-\cite{Song19}. Most of
these schemes have exponential or subexponential subpacketization.
More interestingly, a family of coded caching schemes with linear
subpacketization $($i.e., $F=K)$, were constructed in
\cite{Shanmugam17}, using the Ruzsa-Szem$\acute{\text{e}}$redi
graphs. However, this construction is valid only for sufficiently
large $K$. Another family of linear-subpacketization schemes were
constructed in \cite{Song19} using balanced incomplete block
designs (BIBD), which exists only for some special parameters.

In this paper, we propose a family of centralized coded caching
schemes with polynomial subpacketization. Specifically, for any
$t\in\{1,2,\cdots,K-2\}$ and any integer $n$ such that $0\leq
n\leq t$, we construct a coded caching scheme for any $(K,M,N)$
caching systems with $\frac{M}{N}=\frac{t}{K}$,
$$R=\frac{m}{m-1}\frac{\sum_{i=1}^{m-1}(-1)^{i-1}\binom{m-1}{i}
\binom{K-1-i(\ell-1)}{m-2}}
{\sum_{i=1}^m(-1)^{i-1}\binom{m}{i}\binom{K-1-i(\ell-1)}{m-1}}$$
and
$$F=\frac{K}{m}\sum_{i=1}^m(-1)^{i-1}\binom{m}{i}
\binom{K-1-i(\ell-1)}{m-1},$$ where $m=K-t$ and $\ell=K-m-n+1$,
and we can prove that
$$F\leq K\binom{\left(1-\frac{M}{N}\right)K+n}{n}.$$ In particular,
given $M$, $N$ and integer $n\geq 0$, for any positive integer $K$
such that $m=K\left(1-\frac{M}{N}\right)\geq 2$ is an integer and
$2\leq m\leq K-n$, our construction gives a coded caching scheme
for any $(K,M,N)$ caching system with $F\leq
K\binom{\left(1-\frac{M}{N}\right)K+n}{n}=O(K^{n+1})$. Our
construction is based on a family of subsets of $\mathbb
Z_K=\{0,1,\cdots,K-1\}$, called $(m)_\ell$-bounded subsets of
$\mathbb Z_K$, and can be viewed as a generalization of the
construction in \cite{Maddah-Ali14}.

The rest of this paper is organized as follows. We give a formal
formulation of the centralized coded caching problem in Section
\uppercase\expandafter{\romannumeral 2}. We introduce the bounded
subsets of $\mathbb Z_K$ and discuss their properties in Section
\uppercase\expandafter{\romannumeral 3}. Our construction of coded
caching scheme is presented in Section
\uppercase\expandafter{\romannumeral 4}. Finally, the paper is
concluded in Section \uppercase\expandafter{\romannumeral 5}.

\section{Preliminaries}
For any positive integer $n$, denote $[n]:=\{1,2,\cdots,n\}$. For
any set $X$, $|X|$ is the size (cardinality) of $X$. If
$Y\subseteq X$ and $|Y|=m$, where $0\leq m\leq |X|$, we call $Y$
an $m$-subset of $ X$. We use $\binom{X}{m}$ to denote the
collection of all $m$-subsets of $X$.


We consider a $(K,M,N)$ \emph{caching system}, where one server is
connected by $K$ users through a shared, error-free link. The
server has $N$ files, denoted by
$\textbf{W}_1,\cdots,\textbf{W}_N$, such that each file
$\textbf{W}_i\in\mathbb F^F$ for some fixed finite field $\mathbb
F$. In this paper, we assume that $\mathbb F=\mathbb F_2$, i.e.,
the binary field. Each user $k$ has a local cache memory that
allows it to store a vector $\textbf{Z}_k\in\mathbb F^{MF}$, where
$F$ is referred to as the \emph{subpacketization}.

The caching system operates in two phases: the placement phase and
the delivery phase. In the placement phase, the vector
$\textbf{Z}_k$ is computed and allocated into the cache memory of
each user $k$. In the delivery phase, each user $k$ demands a file
$\textbf{W}_{d_k}$ for some $d_k\in[N]$. The server, having been
informed of the demands of all users, computes a vector
$\textbf{X}_{\textbf{d}}\in\mathbb F^{\lfloor RF\rfloor}$ for some
fixed real number $R$ and transmits it to the users, where
$\textbf{d}=(d_0,d_1,\cdots,d_{K-1})\in[N]^K$ is called the
\emph{demand vector}. An $F$-division coded caching scheme with a
rate $R$ is specified by three sets of functions:
\begin{itemize}
 \item[(i)] (Placement Scheme) a set of caching functions
 $$\left\{\phi_k:\mathbb F^{NF}\rightarrow
 \mathbb F^{MF}\right\}_{k\in\mathbb Z_K},$$
 \item[(ii)] (Delivery Scheme) a set of encoding functions
 $$\left\{\varphi_{\textbf{d}}: \mathbb F^{NF}\rightarrow
 \mathbb F^{\lfloor RF\rfloor}\right\}_{\textbf{d}\in[N]^K},$$
 \item[(iii)] (Decoding Scheme) a set of decoding functions
 $$\left\{\mu_{k,\textbf{d}}: \mathbb F^{MF}
 \times\mathbb F^{\lfloor RF\rfloor}\rightarrow \mathbb
 F^{NF}\right\}_{k\in\mathbb Z_K,\textbf{d}\in[N]^K},$$
\end{itemize}
such that for all $k\in\mathbb Z_K$ and
$\textbf{d}=(d_0,d_1,\cdots,d_{K-1})\in[N]^K$,
$$\textbf{W}_{d_k}=\mu_{k,\textbf{d}}
(\textbf{Z}_k,\textbf{X}_{\textbf{d}}),$$ where
$\textbf{Z}_k=\phi_k(\textbf{W}_1,\cdots,\textbf{W}_N)$ and
$\textbf{X}_{\textbf{d}}=\varphi_{\textbf{d}}
(\textbf{W}_1,\cdots,\textbf{W}_N)$.

Clearly, the decoding scheme is completely determined by the
placement scheme and the delivery scheme. A caching scheme is said
to have \emph{uncoded placement} if $\textbf{Z}_k$ consists of an
exact copy of some subpackets of
$\textbf{W}_1,\cdots,\textbf{W}_N$. Otherwise, it is said to have
\emph{coded placement}.

\section{Bounded Subsets of $\mathbb Z_K$}

In this section, we always assume that $K,m,\ell$ are positive
integers such that $K\geq 2$, $m\leq K$ and $\ell\leq K-m+1$.
Denote $\mathbb Z_K=\{0,1,\cdots,K-1\}$. A family of subsets of
$\mathbb Z_K$, referred to as $(m)_{\ell}$-bounded subsets of
$\mathbb Z_K$, is introduced, which will be used, in the next
section, to construct coded caching schemes with polynomial
subpacketization.

We first give a different representation of the $m$-subsets of
$\mathbb Z_K$. Denote
\begin{align}\label{defeq-VKt}
\nonumber V_K(m)=&\Big\{\Big.(k,a_1,\cdots\!,a_m)\!\in\!\mathbb
Z^{m+1}\!: 0\!\leq \!k\!\leq\! K\!-\!1, a_i\!\geq\! 1\\
&~~\text{for all}~ i\!\in\![m], ~\text{and}~
\sum_{i=1}^m\!a_i\!=\!K\Big.\Big\}.\end{align} For each
$\textbf{v}=(k,a_1,\cdots,a_m)\in V_K(m)$, let
\begin{align}\label{defeq-f}
f(\textbf{v})=\Bigg\{k+\sum_{j=1}^{i-1}a_j~(\text{mod}~ K):
i\in[m-1]\Bigg\}.\end{align} Clearly, $f(\textbf{v})$ is an
$m$-subset of $\mathbb Z_K$, and from \eqref{defeq-f}, we obtain a
mapping $f: V_K(m)\rightarrow\binom{\mathbb Z_K}{m}$. Hence, each
$\textbf{v}\in V_K(m)$ can be used to represent an $m$-subset of
$\mathbb Z_K$.

As an example, consider $K=20$ and $m=5$. Suppose
$\textbf{v}=(12,3,2,6,7,2)$. Then we have $\textbf{v}\in
V_{20}(5)$. By \eqref{defeq-f}, we can obtain
$f(\textbf{v})=\{12,15,17,3,10\}\in\binom{\mathbb Z_{20}}{5}$.

\begin{lem}\label{Prt-f}
Let $f$ be the mapping defined according to \eqref{defeq-f}.
\begin{itemize}
 \item[1)] $f$ is surjective.
 \item[2)] If $A$ is an $m$-subset of $\mathbb Z_K$, then $|f^{-1}(A)|=m$
 and $f^{-1}(A)$ is of the form
 $$f^{-1}(A)=\Big\{\Big.\left(k,a_1^{(A,k)},a_2^{(A,k)},\cdots,
 a_{m}^{(A,k)}\right): k\in A\Big\}\Big.,$$
 where $\left(a_1^{(A,k)},a_2^{(A,k)},\cdots,a_{m}^{(A,k)}\right)$ is
 uniquely determined by $A$ and $k$. Moreover, if $k$
 and $k'$ are two distinct elements of $A$, then
 $\left(a_1^{(A,k')},a_2^{(A,k')},\cdots,a_{m}^{(A,k')}\right)$
 is a circular shift of
 $\left(a_1^{(A,k)},a_2^{(A,k)},\cdots,a_{m}^{(A,k)}\right)$.
\end{itemize}
\end{lem}
\begin{proof}
1) Suppose $A=\{k_1,k_2\cdots,k_{m}\}$ such that
$k_1<k_2<\cdots<k_m$. For each $k=k_{i_0}\in A$, $i_0\in[m]$, let
\begin{equation}\label{eq1-A-to-v}
a_i^{(A,k)}\!=\!\left\{\begin{aligned}
&\!k_{i_0+i}\!-\!k_{i_0+i-1},
{\footnotesize ~}~~~~~~~\text{for}~1\!\leq\! i\!\leq\! m\!-\!i_0,\\
&\!K\!+\!k_{1}\!-\!k_{m},
{\footnotesize ~~~~}~~~~~~~~\text{for}~i\!=\!m\!-\!i_0\!+\!1~,\\
&\!k_{i_0+i-m}\!-\!k_{i_0+i-m-1},\text{for}~m\!-\!i_0\!+\!1\!<\!i\!\leq\!m,\\
\end{aligned} \right.
\end{equation}
and let
\begin{align}\label{eq2-A-to-v}
\textbf{v}_{A,k}=\left(k,a_1^{(A,k)},a_2^{(A,k)},\cdots,
a_m^{(A,k)}\right).\end{align} It is a mechanical work to verify
that $\textbf{v}_{A,k}\in V_K(m)$ and $f(\textbf{v}_{A,k})\!=\!A$,
so $f$ is surjective and $\left\{\textbf{v}_{A,k}\!: k\!\in\!
A\right\}\!\subseteq\! f^{-1}(A)$.

2) According to \eqref{eq1-A-to-v},
$\left(a_1^{(A,k)},a_2^{(A,k)},\cdots,a_m^{(A,k)}\right)$ is
uniquely determined by $A$ and $k$. Moreover, if
$i=i_0+1~(\text{mod}~ m)$ and $k'=k_{i}$, then by
\eqref{eq1-A-to-v}, we can find that
$\left(a_1^{(A,k')},a_2^{(A,k')},\cdots,a_{m}^{(A,k')}\right)
=\left(a_2^{(A,k)},\cdots,a_{m}^{(A,k)},a_{1}^{(A,k)}\right)$ is a
circular shift of
$\left(a_1^{(A,k)},a_2^{(A,k)},\cdots,a_m^{(A,k)}\right)$. Hence,
by induction, for any $k'\in A\backslash\{k\}$,
$\left(a_1^{(A,k')},a_2^{(A,k')},\cdots,a_{m}^{(A,k')}\right)$ is
a circular shift of
$\left(a_1^{(A,k)},a_2^{(A,k)},\cdots,a_{m}^{(A,k)}\right)$.

We now prove that $\left\{\textbf{v}_{A,k}\!: k\!\in\!
A\right\}=f^{-1}(A)$ for all $A\in\binom{\mathbb Z_K}{m}$, where
$\textbf{v}_{A,k}$ is defined by \eqref{eq2-A-to-v}. Since we have
proved $\left\{\textbf{v}_{A,k}\!: k\!\in\! A\right\}\!\subseteq\!
f^{-1}(A)$ and $|\left\{\textbf{v}_{A,k}\!: k\!\in\!
A\right\}|=|A|=m$, it suffices to prove that $|f^{-1}(A)|=m$ for
all $A\in\binom{\mathbb Z_K}{m}$. We can prove this by
contradiction. Suppose $|f^{-1}(A)|>m$ for some
$A\in\binom{\mathbb Z_K}{m}$. Since by 1), $f$ is surjective, then
we have
\begin{align}\label{eq1-Prt-f}
|V_{K}(m)|=\left|\bigcup_{A\in\binom{\mathbb
Z_K}{m}}f^{-1}(A)\right|> m\binom{K}{m}.\end{align} On the other
hand, the number of integer solutions to the equation
$a_1+\cdots+a_m=K$ under the condition that $a_i\geq 1$ for all
$i=1,\cdots,n$, is $\binom{K-1}{m-1}~($e.g., see Chapter 1 of
\cite{Jukna}$)$. So by \eqref{defeq-VKt}, we have
\begin{align*}
|V_{K}(m)|=K\binom{K-1}{m-1}=m\binom{K}{m},\end{align*} which
contradicts to \eqref{eq1-Prt-f}, so it must be the case that
$|f^{-1}(A)|=m$ for all $A\in\binom{\mathbb Z_K}{m}$, and hence,
we have $f^{-1}(A)=\left\{\textbf{v}_{A,k}\!: k\!\in\! A\right\}$
for all $A\in\binom{\mathbb Z_K}{m}$.
\end{proof}

\begin{exam}\label{exm-f}
Let $K\!=\!20$, $m\!=\!5$ and $A\!=\!\{2,3,11,15,19\}$. By
\eqref{eq1-A-to-v}, we have $a_1^{(A,2)}=3-2=1$,
$a_2^{(A,2)}=11-3=8$, $a_3^{(A,2)}=15-11=4$,
$a_4^{(A,2)}=19-15=4$, and $a_5^{(A,2)}=2+20-19=3$. So by
\eqref{eq2-A-to-v}, $\textbf{v}_{A,2}=(2,1,8,4,4,3)$. Similarly,
$\textbf{v}_{A,3}=(3,8,4,4,3,1)$,
$\textbf{v}_{A,11}=(11,4,4,3,1,8)$,
$\textbf{v}_{A,15}=(15,4,3,1,8,4)$ and
$\textbf{v}_{A,19}=(3,3,1,8,4,4)$. By Lemma \ref{Prt-f}, we obtain
$f^{-1}(A)=\{\textbf{v}_{A,2}, \textbf{v}_{A,3},
\textbf{v}_{A,11}, \textbf{v}_{A,15}, \textbf{v}_{A,19}\}$.
Clearly, $(4,4,3,1,8)$ is a circular shift of $(8,4,4,3,1)$. In
fact, for any distinct $k,k'\in A$,
$\left(a_1^{(A,k')},a_2^{(A,k')},\cdots,a_{5}^{(A,k')}\right)$ is
a circular shift of
$\left(a_1^{(A,k)},a_2^{(A,k)},\cdots,a_{5}^{(A,k)}\right)$.\end{exam}

By Lemma \ref{Prt-f}, each $m$-subset $A$ of $\mathbb Z_K$ can be
represented by a subset $f^{-1}(A)$ of $V_K(m)$. Now, we can
introduce the concept of $(m)_{\ell}$-bounded subset of $\mathbb
Z_K$. Denote
\begin{align}\label{defeq-VKlt} \nonumber
V_{K,\ell}(m)=&\Big\{\Big.(k,a_1,\cdots\!,a_m)\in V_K(m): a_i\geq
\ell \\&~\text{for some}~ i\in[m]\Big.\Big\}.\end{align}
\begin{defn}\label{def-BNDS}
An $m$-subset $A$ of $\mathbb Z_K$ is called an
$(m)_{\ell}$-\emph{bounded subset} of $\mathbb Z_K$ if
$f^{-1}(A)\cap V_{K,\ell}(m)\neq\emptyset$. Let $\mathcal
B_{K,\ell}(m)$ denote the collection of all $(m)_{\ell}$-bounded
subsets of $\mathbb Z_K$.\end{defn}


\begin{rem}\label{rem-BNDS}
We point out two simple facts about the $(m)_{\ell}$-bounded
subset of $\mathbb Z_K$.
\begin{itemize}
 \item[1)] For any $A\in\binom{\mathbb Z_K}{m}$, if $f^{-1}(A)\cap
 V_{K,\ell}(m)\neq\emptyset$, then $f^{-1}(A)\subseteq
 V_{K,\ell}(m)$. Hence, $A$ is an $(m)_{\ell}$-bounded subset of
 $\mathbb Z_K$ if and only if $f^{-1}(A)\subseteq
 V_{K,\ell}(m)$. In fact, for
 any distinct $k,k'\in A$, by Lemma \ref{Prt-f},
 $\left(a_1^{(A,k')},a_2^{(A,k')},\cdots,a_{m}^{(A,k')}\right)$ is
 a circular shift of
 $\left(a_1^{(A,k)},a_2^{(A,k)},\cdots,a_{m}^{(A,k)}\right)$,
 so by \eqref{defeq-VKlt}, if
 $\textbf{v}_{A,k}\in V_{K,\ell}(m)$, then $\textbf{v}_{A,k'}\in
 V_{K,\ell}(m)$. In other words, if $f^{-1}(A)\cap
 V_{K,\ell}(m)\neq\emptyset$, then $f^{-1}(A)\subseteq
 V_{K,\ell}(m)$.
 \item[2)] If $\ell<\frac{K}{m}+1$, then $\mathcal
 B_{K,\ell}(m)=\binom{\mathbb Z_K}{m}$. This can be proved as follows.
 For any
 $A\in\binom{\mathbb Z_K}{m}$ and $(k,a_1,\cdots,a_m)\in
 f^{-1}(A)$, we always have $a_i\geq \ell$ for some $i\in[m]$.
 $($Otherwise we can obtain $\sum_{i=1}ta_i\leq
 m(\ell-1)<m\frac{K}{m}=K$, which contradicts to
 \eqref{defeq-VKt}.$)$ Hence, by \eqref{defeq-VKlt},
 $A\in\mathcal B_{K,\ell}(m)$. Since $A\in\binom{\mathbb Z_K}{m}$
 is arbitrary, then we have $\mathcal B_{K,\ell}(m)=\binom{\mathbb
 Z_K}{m}$.
\end{itemize}
\end{rem}

Let's reconsider Example \ref{exm-f}. We can verify that
$f^{-1}(A)\subseteq V_{20,8}(5)$, where $A=\{2,3,11,15,19\}$, so
$A$ is a $(5)_8$-bounded subset of $\mathbb Z_{20}$. We can
further consider the $4$-subset $B=\{2,3,11,19\}$ of $A$. By
\eqref{eq1-A-to-v} and \eqref{eq2-A-to-v}, we have
$\textbf{v}_{B,2}=(2,1,8,8,3)\in f^{-1}(B)$. By \eqref{eq1-Prt-f},
$\textbf{v}_{B,2}\in V_{20,8}(4)$, so $B$ is a $(4)_8$-bounded
subset of $\mathbb Z_{20}$. What is interesting in this example is
that $a_1^{(B,2)}=a_1^{(A,2)}$, $a_2^{(B,2)}=a_2^{(A,2)}$,
$a_3^{(B,2)}=a_3^{(A,2)}+a_4^{(A,2)}$ and
$a_4^{(B,2)}=a_5^{(A,2)}$. In fact, this holds for all $t$-subset
$A$ of $\mathbb Z_K$ and all $(t-1)$-subset $B$ of $A$. In
general, we have the following lemma.

\begin{lem}\label{Subset-BS}
Suppose $2\leq m\leq K$ and $A$ is an $(m)_{\ell}$-bounded subset
of $\mathbb Z_K$. Then any $(m-1)$-subset of $A$ is an
$(m-1)_{\ell}$-bounded subset of $\mathbb Z_K$.
\end{lem}
\begin{proof}
Suppose $A=\{k_1,k_2\cdots,k_{m}\}$ such that $0\leq
k_1<k_2<\cdots<k_m\leq K-1$, and $B=A\backslash\{k_{i_0}\}$, where
$i_0\in[m]$. Let $i_1=i_0+1 ~(\text{mod}~m)$. By
\eqref{eq1-A-to-v} and \eqref{eq2-A-to-v}, we can verify that
$\textbf{v}_{B,k_{i_1}}=\left(k_{i_1},a_1^{(B,k_{i_1})},\cdots,
a_{m-1}^{(B,k_{i_1})}\right)=\left(k_{i_1},a_1^{(A,k_{i_1})},\cdots,
a_{m-2}^{(A,k_{i_1})},
a_{m-1}^{(A,k_{i_1})}+a_{m}^{(A,k_{i_1})}\right)$.

By 2) of Lemma \ref{Prt-f}, $\textbf{v}_{A,k_{i_1}}\in f^{-1}(A)$.
Since $A$ is a $(m)_{\ell}$-bounded subset of $\mathbb Z_K$, we
have $a_{i}^{(A,k_{i_1})}\geq\ell$ for some $i\in[m]$, and so
$a_{i'}^{(B,k_{i_1})}\geq\ell$ for some $i'\in[m-1]$. By
\eqref{defeq-VKlt}, we have $\textbf{v}_{B,k_{i_1}}\in
V_{K,\ell}(m-1)$. Moreover, by 2) of Lemma \ref{Prt-f},
$\textbf{v}\in f^{-1}(B)$, and so $f^{-1}(B)\cap
V_{K,\ell}(m-1)\neq\emptyset$. Hence, $B$ is an
$(m-1)_{\ell}$-bounded subset of $\mathbb Z_K$.
\end{proof}

The following lemma counts the number of $(m)_{\ell}$-bounded
subsets of $\mathbb Z_K$.

\begin{lem}\label{P-in-BS}
Suppose $K,m,\ell$ are positive integers such that $K\geq 2$,
$m\leq K$ and $\ell\leq K-m+1$. We have
\begin{itemize}
 \item[1)] For each $k\in\mathbb Z_K$, the
 number of $(m)_{\ell}$-bounded subsets of $\mathbb Z_K$
 containing $k$, denoted by $C(K,m,\ell)$, is independent of $k$,
 and we have
 $$C(K,m,\ell)=\sum_{i=1}^m(-1)^{i-1}\binom{m}{i}\binom{K-1-i(\ell-1)}{m-1}.$$
 \item[2)] The number of $(m)_{\ell}$-bounded subsets of
 $\mathbb Z_K$ is
 $$|\mathcal B_{K,\ell}(m)|=
 \frac{K}{m}\sum_{i=1}^m(-1)^{i-1}\binom{m}{i}\binom{K-1-i(\ell-1)}{m-1}.$$
 \item[3)] The number of $(m)_{\ell}$-bounded subsets of
 $\mathbb Z_K$ satisfies
 \begin{align}\label{eq-Est-BKLM}|\mathcal B_{K,\ell}(m)|\leq
 K\binom{K-\ell+1}{m}.\end{align}
\end{itemize}
\end{lem}
\begin{proof}
1) For $k\in \mathbb Z_K$, let $\mathcal S_K(k)$ denote the
collection of all $m$-subsets of $\mathbb Z_K$ containing $k$.
Clearly, $$|\mathcal S_K(k)|=\binom{K-1}{m-1}.$$ Let $\mathcal
T_K(k)$ denote the collection of all $m$-subsets of $\mathbb Z_K$
that contain $k$ but are not an $(m)_{\ell}$-bounded subsets of
$\mathbb Z_K$.  We now compute $|\mathcal T_K(k)|$. If
$A\in\mathcal T_K(k)$, by 1) of Remark \ref{rem-BNDS},
$\left(k,a_1^{(A,k)},\cdots, a_m^{(A,k)}\right)\notin
V_{K,\ell}(m)$, so we obtain an $m$-tuple
$\left(a_1^{(A,k)},\cdots, a_m^{(A,k)}\right)\in\mathbb Z_K^m$
satisfying $\sum_{i=1}^ma_i^{(A,k)}=K$ and $1\leq a_i^{(A,k)}\leq
\ell-1$ for all $i\in[m]$. Conversely, for any $m$-tuple
$\left(a_1,\cdots, a_m\right)\in\mathbb Z_K^m$ satisfying
$\sum_{i=1}^ma_i=K$ and $1\leq a_i\leq \ell-1$ for all $i\in[m]$,
by \eqref{defeq-f}, we have $f(\textbf{v})\in\mathcal T_K(k)$,
where $\textbf{v}=(k,a_1,\cdots,a_m)\in V_{K}(m)$. Hence,
$|\mathcal T_K(k)|$ equals to the number of $m$-tuples
$(a_1,\cdots,a_m)\in\mathbb Z_K^m$ satisfying $\sum_{i=1}^ma_i=K$
and $1\leq a_i\leq \ell-1$ for all $i\in[m]$. By letting
$x_i=a_i-1$ for each $i\in[m]$, we can further show that
$|\mathcal T_K(k)|=\omega_{m,\ell-1}(K-m)$, where
$\omega_{m,\ell-1}(K-m)$ denotes the number of $m$-tuples
$(x_1,\cdots,x_m)\in\mathbb Z_K^m$ satisfying
$\sum_{i=1}^mx_i=K-m$ and $0\leq x_i<\ell-1$ for all $i\in[m]$. By
\cite[Lemma 1.1]{Ratsaby},
$\omega_{m,\ell-1}(K-m)=\sum_{i=0}^m(-1)^i
\binom{m}{i}\binom{K-1-i(\ell-1)}{m-1}$, so we have
\begin{align*}|\mathcal
T_K(k)|=\sum_{i=0}^m(-1)^i
\binom{m}{i}\binom{K-1-i(\ell-1)}{m-1}.\end{align*} Thus, the
number of $(m)_{\ell}$-bounded subsets of $\mathbb Z_K$ containing
$k$ equals to
\begin{align*}
&~|\mathcal S_K(k)\backslash\mathcal T_K(k)|\\&=|\mathcal
S_K(k)|-|\mathcal T_K(k)|\\&=\binom{K-1}{m-1}-\sum_{i=0}^m(-1)^i
\binom{m}{i}\binom{K-1-i(\ell-1)}{m-1}\\
&=\sum_{i=1}^m(-1)^{i-1}\binom{m}{i}\binom{K-1-i(\ell-1)}{m-1},
\end{align*} which proves claim 1).

2) By claim 1), for each $k\in\mathbb Z_K$, the set of
$(m)_{\ell}$-bounded subsets of $\mathbb Z_K$ containing $k$ is
$C(K,m,\ell)=\sum_{i=1}^m(-1)^{i-1}\binom{m}{i}
\binom{K-1-i(\ell-1)}{m-1}$, which is independent of $k$. On the
other hand, by Definition \ref{def-BNDS}, each
$(m)_{\ell}$-bounded subset of $\mathbb Z_K$ is an $m$-subset of
$\mathbb Z_K$. Then by counting the $1$s in the incidence matrix
of $\mathcal B_{K,\ell}(m)$, we have
$$KC(K,m,\ell)=|\mathcal B_{K,\ell}(m)|m.$$ Thus, the total number
of $(m)_{\ell}$-bounded subsets of $\mathbb Z_K$ is
\begin{align*}
|\mathcal B_{K,\ell}(m)|&=\frac{KC(K,m,\ell)}{m}\\
&=\frac{K}{m}\sum_{i=1}^m(-1)^{i-1}\binom{m}{i}
\binom{K-1-i(\ell-1)}{m-1},\end{align*} which proves 2).

3) For each $k\in\mathbb Z_K$, denote
$$X_{(k,\ell)}=\{k,k\oplus_K1,\cdots,k\oplus_K(K-\ell),$$ where
$k\oplus_Ki=k+i~(\text{mod}~K)$ for any $i\in[K-\ell]$. Note that
$|X_{(k,\ell)}|=K-\ell+1$. We are to prove that if $A$ is an
$(m)_{\ell}$-bounded subset of $\mathbb Z_K$, then
$A\in\binom{X_{(k,\ell)}}{m}$ for some $k\in\mathbb Z_K$. In fact,
suppose $A=\{k_1,\cdots,k_m\}$ such that $0\leq k_1<\cdots<k_m\leq
K-1$. Since $A$ is an $(m)_{\ell}$-bounded subset of $\mathbb
Z_K$, by 1) of Remark \ref{rem-BNDS},
$\textbf{v}_{A,k_j}=\left(k_j,a_1^{(A,k_j)},\cdots,
a_m^{(A,k_j)}\right)\in V_{K,\ell}(m)$ for all $j\in[m]$, so
$a_{i}^{(A,k_j)}\geq \ell$ for some $i\in[m]$. Then by
\eqref{eq1-A-to-v}, $k_{i'+1}-k_{i'}\geq\ell$ for some
$i'\in[m]~($If $i'=m$, then $k_{1}+K-k_{m}\geq\ell.)$, and so we
have $A\subseteq X_{(k_{i'+1},\ell)}~($see Example
\ref{exm-Xm-set} for an illustration$)$. Thus, we have $\mathcal
B_{K,\ell}(m)\subseteq\bigcup_{k\in\mathbb
Z_K}\binom{X_{(k,\ell)}}{m}$, and so
\begin{align*}
|\mathcal B_{K,\ell}(m)|\leq\sum_{k\in\mathbb
Z_K}\left|\binom{X_{(k,\ell)}}{m}\right|=K\binom{K-\ell+1}{m},
\end{align*} which proves 3).
\end{proof}

\begin{exam}\label{exm-Xm-set}
Suppose $K=20$, $m=5$ and $\ell=8$. Let
$A=\{k_1,k_2,k_3,k_4,k_5\}=\{1,4,13,14,18\}$, where $k_1=1$,
$k_2=4$, $k_3=13$, $k_4=14$ and $k_5=18$. By \eqref{eq1-A-to-v},
we can obtain $\textbf{v}_{A,1}=(1,3,9,1,4,3)$, so $A\in\mathcal
B_{20,9}(5)$. Note that by \eqref{eq1-A-to-v},
$a_5^{(A,2)}=9=k_3-k_2$, and we can verify that $A\subseteq
X_{(k_3,\ell)}=\{k_3,k_3\oplus_K1,\cdots,k_3\oplus_K(K-\ell)\}=
\{13,14,\cdots,19,1,2,3,4\}$.
\end{exam}

\section{Coded Caching With Polynomial Subpacketization}

In this section, we construct a family of coded caching schemes
using the $(m)_{\ell}$-bounded subsets of $\mathbb Z_K$.

Suppose $K$, $m$ and $\ell$ are positive integers such that $2\leq
m\leq K-1$ and $\ell\leq K-m+1$. We use $\mathbb Z_K$ to denote
the set of $K$ users, and each file $\textbf{W}_n$ is divided into
$F=|\mathcal B_{K,\ell}(m)|$ packets. $($Note that $\mathcal
B_{K,\ell}(m)$ is the collection of all $(m)_{\ell}$-bounded
subsets of $\mathbb Z_K$.$)$ Then we can denote
\begin{align}\label{CCS1-FP}
\textbf{W}_n=\{W_{n,S}\in\mathbb F_2: S\in\mathcal
B_{K,\ell}(m)\}.\end{align} Moreover, for each $T\in\mathcal
B_{K,\ell}(m-1)$, denote
\begin{align}\label{CCS1-UT}
\mathcal U(T)=\{k\in\mathbb Z_K: (T\cup\{k\})\in\mathcal
B_{K,\ell}(m)\},\end{align} and for each $k\in\mathbb Z_K$, denote
\begin{align}\label{CCS1-Vk} \mathcal V(k)=\{T\in\mathcal
B_{K,\ell}(m-1): (T\cup\{k\})\in\mathcal
B_{K,\ell}(m)\}.\end{align}

Now, we have the following construction.

\textbf{Construction 1}: A coded caching scheme is as follows.
\begin{itemize}
 \item[(i)] (Placement Scheme) For each $k\in\mathbb Z_K$,
 the user $k$ caches
 \begin{align}\label{CCS1-PS}\textbf{Z}_k\!=\!\{W_{n,S}: n\!\in\![N],
 S\!\in\!\mathcal B_{K,\ell}(m)~\text{and}~k\!\notin\! S\}.\end{align}
 \item[(ii)] (Delivery Scheme) Given any
 $\textbf{d}=(d_0,d_1,\cdots,d_{K-1})\in[N]^K$, for each $T\in\mathcal
 B_{K,\ell}(m-1)$, the server transmits
 \begin{align}\label{CCS1-DS}
 X_T=\oplus_{k\in \mathcal U(T)}W_{d_k,T\cup\{k\}},\end{align} where
 $\oplus$ denotes the bitwise XOR. 
 \item[(iii)] (Decoding Scheme) Given any
 $\textbf{d}=(d_0,d_1,\cdots,d_{K-1})\in[N]^K$, for each $k\in\mathbb Z_K$
 and each $T\in \mathcal V(k)$,
 \begin{align}\label{CCS1-DC}W_{d_k,T\cup\{k\}}=\left(\oplus_{k'\in
 \mathcal U(T)\backslash\{k\}}W_{d_{k'},T\cup\{k'\}}\right)\oplus
 X_T.\end{align}
\end{itemize}

Clearly, the decoding equality \eqref{CCS1-DC} can be derived
directly from \eqref{CCS1-DS}. We still have to prove that each
user can recover its requested file by the decoding scheme.
\begin{lem}\label{lem-CCS1-DC}
In Construction 1, for each $k\in\mathbb Z_K$, the user $k$ can
successfully recover its requested file $\textbf{W}_{d_k}$.
\end{lem}
\begin{proof}
By \eqref{CCS1-FP} and \eqref{CCS1-PS}, it suffices to prove that
for each $k\!\in\!\mathbb Z_K$ and $S\!\in\!\mathcal
B_{K,\ell}(m)$ such that $k\!\in\! S$, the user $k$ can recover
$W_{d_k,S}$ from its cached packets and received packets.

Let $T=S\backslash\{k\}$. By Lemma \ref{Subset-BS}, we have
$T\in\mathcal B_{K,\ell}(m-1)$, $T\in \mathcal V(k)$ and $k\in
\mathcal U(T)$, where $\mathcal V(k)$ and $\mathcal U(T)$ are
defined as in \eqref{CCS1-Vk} and \eqref{CCS1-UT}, respectively.
For each $k'\in \mathcal U(T)\backslash\{k\}$, since
$T=S\backslash\{k\}$, we have $k\notin T\cup\{k'\}$. Moreover, by
\eqref{CCS1-UT}, we have $T\cup\{k'\}\in\mathcal B_{K,\ell}(m)$.
Then by \eqref{CCS1-PS}, the user $k$ caches
$W_{d_{k'},T\cup\{k'\}}$ for each $k'\in \mathcal
U(T)\backslash\{k\}$, and hence it can recover
$W_{d_k,T\cup\{k\}}=W_{d_k,S}$ by \eqref{CCS1-DC}.
\end{proof}

\begin{thm}\label{thm-main1}
Construction 1 gives a coded caching scheme for any $(K,M,N)$
caching system with $\frac{M}{N}=1-\frac{m}{K},$
\begin{align*}
F=\frac{K}{m}\sum_{i=1}^m(-1)^{i-1}\binom{m}{i}\binom{K-1-i(\ell-1)}{m-1},
\end{align*}
and
$$R=\frac{m}{m-1}\frac{\sum_{i=1}^{m-1}(-1)^{i-1}\binom{m-1}{i}
\binom{K-1-i(\ell-1)}{m-2}}
{\sum_{i=1}^m(-1)^{i-1}\binom{m}{i}\binom{K-1-i(\ell-1)}{m-1}}.$$
Moreover, denoting $n=K-m+1-\ell$, then
\begin{align*}
F\leq K\binom{\left(1-\frac{M}{N}\right)K+n}{n}.
\end{align*}
\end{thm}
\begin{proof}
By Lemma \ref{lem-CCS1-DC}, Construction 1 is a coded caching
scheme for any $(K,M,N)$ caching system with $K$ users and $N$
files, and we have seen that each file is divided into
$F=|\mathcal B_{K,\ell}(m)|
=\frac{K}{m}\sum_{i=1}^m(-1)^{i-1}\binom{m}{i}\binom{K-1-i(\ell-1)}{m-1}$
packets. 

For each $k\in\mathbb Z_K$, by \eqref{CCS1-PS}, each user caches
$|\mathcal B_{K,\ell}(m)|-C(K,m,\ell)$ packets of each file, where
$C(K,m,\ell)$ is the number of $(m)_{\ell}$-bounded subsets of
$\mathbb Z_K$ containing $k$. In the proof of 2) of Lemma
\ref{P-in-BS}, we have seen that $C(K,m,\ell)=\frac{|\mathcal
B_{K,\ell}(m)|m}{K}$, so we can obtain
\begin{align*}\frac{M}{N}&
=\frac{|\mathcal B_{K,\ell}(m)|-C(K,m,\ell)}{F}\\
&=\frac{|\mathcal B_{K,\ell}(m)|-\frac{|\mathcal
B_{K,\ell}(m)|m}{K}}{|\mathcal B_{K,\ell}(m)|}\\
&=1-\frac{m}{K}.\end{align*}

By the delivery scheme of Construction 1, the total number of
packets transmitted by the server is $RF\!=\!|\mathcal
B_{K,\ell}(m\!-\!1)|$, so
\begin{align*}R&=\frac{|\mathcal B_{K,\ell}(m-1)|}{F}\\
&=\frac{\frac{K}{m-1}\sum_{i=1}^{m-1}(-1)^{i-1}\binom{m-1}{i}
\binom{K-1-i(\ell-1)}{m-2}}
{\frac{K}{m}\sum_{i=1}^m(-1)^{i-1}\binom{m}{i}\binom{K-1-i(\ell-1)}{m-1}}\\
&=\frac{m}{m-1}\frac{\sum_{i=1}^{m-1}(-1)^{i-1}\binom{m-1}{i}
\binom{K-1-i(\ell-1)}{m-2}}
{\sum_{i=1}^m(-1)^{i-1}\binom{m}{i}\binom{K-1-i(\ell-1)}{m-1}}.
\end{align*}

Moreover, noticing that $\frac{M}{N}=1-\frac{m}{K}$, we can obtain
$m=K\left(1-\frac{M}{N}\right)$. 
Since $n=K-m+1-\ell$, then
$K-\ell+1=m+n=\left(1-\frac{M}{N}\right)K+n$. So by 3) of Lemma
\ref{P-in-BS}, we have
\begin{align*}|\mathcal B_{K,\ell}(m)|&\leq
K\binom{K-\ell+1}{m}\\&=K\binom{K-\ell+1}{K-\ell+1-m}\\
&=K\binom{\left(1-\frac{M}{N}\right)K+n}{n},\end{align*} which
completes the proof.
\end{proof}

We can compare our construction with the Maddah-Ali-Niesen scheme
\cite{Maddah-Ali14}. For any $t\in[K-2]$ and any integer $n$ such
that $0\leq n\leq t$, let $m=K-t$ and $\ell=K-m+1-n$. Then from
Construction 1, we obtain a coded caching scheme for any $(K,M,N)$
caching system with $\frac{M}{N}=1-\frac{m}{K}=\frac{t}{K}$ and
$F\leq K\binom{\left(1-\frac{M}{N}\right)K+n}{n}$. Moreover, we
have
\begin{itemize}
 \item[1)] For $n>t-\frac{K}{K-t}$, we have
 $\ell<\frac{K}{m}+1$, and by 2) of Remark \ref{rem-BNDS}, $\mathcal
 B_{K,\ell}(m)=\binom{\mathbb Z_K}{m}$ and $\mathcal
 B_{K,\ell}(m-1)=\binom{\mathbb Z_K}{m-1}$. By Theorem \ref{thm-main1},
 it can be verified that the caching scheme
 obtained from Construction 1 has $F=\binom{K}{KM/N}$ and
 $R=\frac{K(1-M/N)}{1+KM/N}$, which are the same as the Maddah-Ali-Niesen
 scheme \cite{Maddah-Ali14}.
 \item[2)] As $n$ decreases, $\ell$ increases and
 by Theorem \ref{thm-main1}, $F$ decreases while $R$ increases. As an
 example, the $\log (F)$ versus $n+1$ and the $R$
 versus $n+1$ for a system with $K=50$ and
 $\frac{M}{N}=\frac{1}{2}$
 are shown in Fig. \ref{fg-RF-n}.
\end{itemize}
\renewcommand\figurename{Fig}
\begin{figure}[htbp]
\begin{center}
\includegraphics[height=3.5cm]{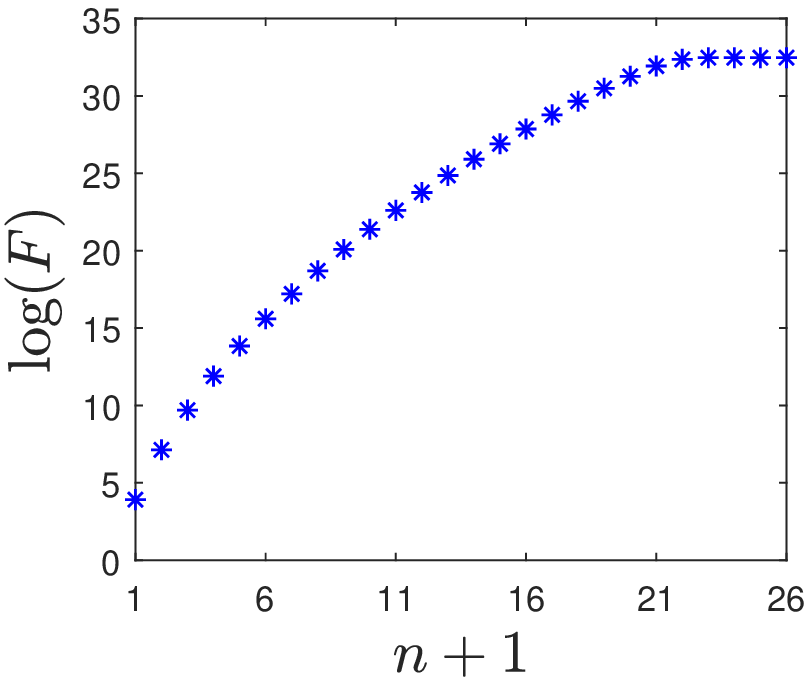}
\includegraphics[height=3.5cm]{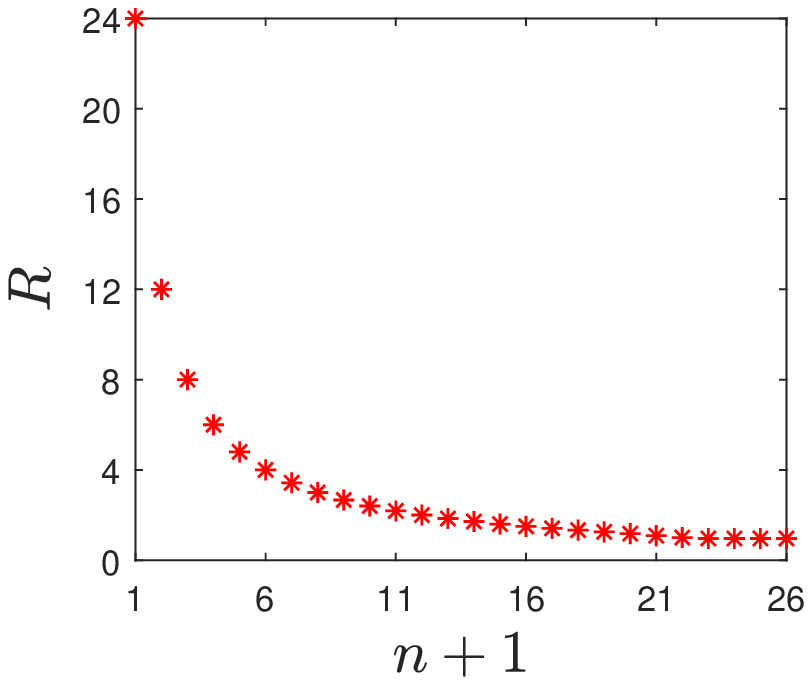}
\end{center}
\vspace{-0.2cm}\caption{The $\log (F)$ versus $n+1$ and the $R$
versus $n+1$ figure for a caching system with $K=50$ and
$\frac{M}{N}=\frac{1}{2}$, where we can obtain $1\leq n+1\leq
26$.}\label{fg-RF-n}
\end{figure}

Construction 1 gives a family of caching schemes with polynomial
subpacketization, as stated by the following theorem.
\begin{thm}\label{thm-main2}
Given an integer $n\geq 0$, for any $K$ such that
$m=K\left(1-\frac{M}{N}\right)$ is an integer and $2\leq m\leq
K-n$, there exists a coded caching scheme for any $(K,M,N)$
caching system with $F\leq
K\binom{\left(1-\frac{M}{N}\right)K+n}{n}=O(K^{n+1})$ and
$R=\frac{m}{m-1}\frac{\sum_{i=1}^{m-1}(-1)^{i-1}\binom{m-1}{i}
\binom{K-1-i(\ell-1)}{m-2}}
{\sum_{i=1}^m(-1)^{i-1}\binom{m}{i}\binom{K-1-i(\ell-1)}{m-1}},$
where $\ell=K-m+1-n$.
\end{thm}
\vspace{3pt}\begin{proof} By assumption, we have $2\leq\! m\leq\!
K-1$
and $1\leq\!\ell\leq\! K\!-m+1$. Therefore, 
Construction 1 gives a coded caching scheme for any $(K,M,N)$
caching system with $F\leq
K\binom{\left(1-\frac{M}{N}\right)K+n}{n}$ and
$R=\frac{m}{m-1}\frac{\sum_{i=1}^{m-1}(-1)^{i-1}\binom{m-1}{i}
\binom{K-1-i(\ell-1)}{m-2}}
{\sum_{i=1}^m(-1)^{i-1}\binom{m}{i}\binom{K-1-i(\ell-1)}{m-1}}$.
\end{proof}

\section{Conclusions}
We construct a family of coded caching schemes, which includes the
schemes with optimal rate as well as the schemes with polynomial
subpacketization. Like all existing constructions, our method
reduces the subpacketization at the cost of increasing the rate.
It is still an open problem to characterize the tight bound on the
rate for coded caching with polynomial subpacketization.

\end{document}